\documentclass[conference]{IEEEtran}
\ifCLASSINFOpdf
\else
\fi
\usepackage[cmex10]{amsmath}
\newcommand{\resolved}[1]{}

\newcommand{\eat}[1]{}
\newcommand{\schb}[1]{\textsf{\small{#1}}} 

\def\q{{t}}                     

\newcommand{\eop}{\hspace*{\fill}\mbox{$\Box$}}     
\newcounter{example}
\renewcommand{\theexample}{\arabic{example}}

\newcounter{definition}
\renewcommand{\thedefinition}{\arabic{definition}}

\newcommand{\mypara}[1]{\noindent\textbf{#1.}}

\newcommand{\map}{\mathcal{M}}
\newcommand{\cand}{\mathcal{C}}
\newcommand{\kc}{K_{\mathcal{C}}}

\newcommand{\expf}{\mathsf{explains_{full}}}
\newcommand{\errf}{\mathsf{error_{full}}}
\newcommand{\creates}{\mathsf{creates}}

\newcommand{\covers}{\mathsf{covers}}

\newcommand{\gold}{\mathcal{M}_G}
\newcommand{\goldmap}{$\mathcal{M}_G$}
\newcommand{\errmap}{\candmap$-$\goldmap} 
\newcommand{\candmap}{$\mathcal{C}$}

\newcommand{\jgold}{K_G}

\newcommand{\tgd}{st~tgd}

\newcommand{\pcorr}{\pi_{\mathit{Corresp}}}
\newcommand{\perr}{\pi_{\mathit{Errors}}}
\newcommand{\pune}{\pi_{\mathit{Unexplained}}}
\usepackage{outlines} 
\usepackage[normalem]{ulem} 
\usepackage[usenames,dvipsnames]{color} 
\usepackage{mydefs} 
\usepackage{booktabs}
\usepackage[labelformat=simple]{subcaption} 
\usepackage{tikz,pgfplots}      
\usetikzlibrary{calc,shapes.geometric,positioning,patterns}
\usepackage{xfrac}              

\begin{document}
%
\title{A Collective, Probabilistic Approach \\to Schema Mapping: Appendix}


\author{\IEEEauthorblockN{Angelika Kimmig}
\IEEEauthorblockA{KU Leuven\\
angelika.kimmig@cs.kuleuven.be}
\and
\IEEEauthorblockN{Alex Memory}
\IEEEauthorblockA{University of Maryland\\
memory@cs.umd.edu}
\and
\IEEEauthorblockN{Ren\'ee  J. Miller}
\IEEEauthorblockA{University of Toronto\\
  miller@cs.toronto.edu}
\and
\IEEEauthorblockN{Lise Getoor}
\IEEEauthorblockA{UC Santa Cruz\\
getoor@ucsc.edu}}


%


\maketitle



%
\IEEEpeerreviewmaketitle

In this appendix we provide additional supplementary material to ``A Collective, Probabilistic Approach to Schema Mapping''~\cite{thispaper}.  We include an additional extended example, supplementary experiment details, and proof for the complexity result stated in the main paper.

\section{Example of Selection over ST TGDs}\label{sec:objective-example}
We extend the running example from the main paper to illustrate objective Eq.~(9) of \cite{thispaper}.  We use a reduced candidate set $\cand'=\{\theta_1,\theta_3\}$ (Figure~1(d) in \cite{thispaper}) and the data in Figure~1(b)-(c) in \cite{thispaper}, but omit the \schb{leader} relation.
A universal solution $K_{\theta_1}$ of $I$ contains
the \schb{task} tuples \schb{(BigData, Bob, Null$_1$)} and \schb{(ML, Alice, Null$_2$)}, 
while a $K_{\theta_3}$ contains
the \schb{task} tuples \schb{(BigData, Bob, Null$_3$)} and \schb{(ML, Alice, Null$_4$)} and the \schb{org} tuples \schb{(Null$_3$, IBM)} and \schb{(Null$_4$, SAP)}.

For $\theta_1$, $\creates$ is $1$ for tuple \schb{task(BigData, Bob, Null$_1$)}, and $0$ for all other tuples, and $\covers$ is $\sfrac{2}{3}$ for \schb{task(ML, Alice, 111)} and $0$ otherwise. This is because \schb{task(ML, Alice, Null$_2$)}  
partially explains the latter via a homomorphism mapping \schb{Null$_2$} to $111$. Similarly, for $\theta_3$, $\creates$ is $1$ for \schb{task(BigData, Bob, Null$_3$)} and \schb{org(Null$_3$,IBM)}, but $0$ for \schb{task(ML, Alice, Null$_4$)} and \schb{org(Null$_4$,SAP)}, which 
partially explain \schb{task(ML, Alice, 111)} and \schb{org(111, SAP)} to degree $\sfrac{3}{3}$ and $\sfrac{2}{2}$ respectively, via a homomorphism mapping \schb{Null$_4$} to $111$, with corresponding values for $\covers$. 
The different subsets of candidate \tgd{}s thus obtain the following values for the individual parts and the total of objective function~Eq.~(9) of \cite{thispaper}.

\begin{center}
 \begin{tabular}{ccccc}
  \hline
  $\map$ & $\sum 1-\mathrm{explains}$ & $\sum \mathrm{error}$ & size & Eq.~(9) of \cite{thispaper} \\ \hline
   $\{ \}$ & 4 & 0 & 0 & 4\\
  $\{\theta_1\}$ & $3\sfrac{1}{3}$ & 1 & 3 & $7\sfrac{1}{3}$ \\
  $\{\theta_{3}\}$ & 2 & 2 & 4 & 8 \\
  $\{\theta_{1},\theta_3\}$ & 2 & 3 & 7 & 12 \\ \hline
  \end{tabular}
\end{center}
As the data example is small compared to the mappings, the minimal value for the objective is that of the empty mapping, but we also see that $\{\theta_1\}$ is preferred over $\{\theta_{3}\}$, which in turn is preferred over $\{\theta_{1},\theta_3\}$. The reason is that while $\theta_3$ covers more tuples than $\theta_1$, it also produces more errors and is larger. The fact that the empty mapping has a better objective value is an important guard against overfitting on too little data; this is easily overcome by slightly larger data instances.
If we add at least five more projects \schb{X} of the same kind as the \schb{ML} one, i.e., pairs of tuples \schb{proj(X,N,1)} and \schb{task(X,Alice,111)}, the preferred mapping is $\{\theta_{3}\}$, as the empty mapping cannot explain the new target tuples, $\theta_1$ explains each to degree $\sfrac{2}{3}$, and  $\theta_3$ fully explains them (while no mapping introduces additional errors).

\section{Scenario generation}
We provide additional details of the scenario generation process discussed in Section VI-A of \cite{thispaper}. 

\mypara{iBench}\label{app:param}
We used seven iBench primitives~\cite{alexe:pvldb08,AGCM15}:
CP copies a source relation to the target, changing its name.  ADD copies a source relation and adds attributes; DL does the same, but removes attributes instead; and ADL adds and removes attributes to the same relation.  The number that are added or removed are controlled by range parameters, which we set to (2,4).  ME copies two relations, after joining them, to form a target relation. VP copies a source relation to form two, joined, target relations. VNM is the same as VP but introduces an additional target relation to form a N-to-M relationship between the other target relations.

\mypara{Modifying the metadata evidence through random correspondences} 
If $\pcorr > 0$ (cf. Table~I of~\cite{thispaper}), we introduce additional correspondences as follows. 
We randomly select $\pcorr$ percent of the target relations.  For every selected target relation~$T$, we randomly select a source relation~$S$ from those of the iBench primitive invocations not involving~$T$ (so Clio~\cite{FHHMPV09} can generate $\gold$ as part of $\cand$). For each attribute of~$T$, we introduce a correspondence to a randomly selected attribute of~$S$.

\mypara{Modifying the data instance} 
As certain errors and certain unexplained tuples can be removed prior to optimization (cf.~Section III-C of \cite{thispaper}), we restrict data instance modifications to non-certain errors and non-certain unexplained tuples (with respect to $\gold$). 
Note that in our scenarios, $\gold\subseteq\cand$, and thus $\jgold \subseteq \kc$. So each tuple in $\kc$ is either generated by both \goldmap\ and \errmap, only by \goldmap\ (i.e., a non-certain error tuple if deleted from $J$), or only by \errmap\ (i.e., a non-certain unexplained tuple if added to $J$). As tuples in $\kc$ may have nulls, we take into account homomorphisms when determining which of these cases applies to a given tuple. 
We randomly select $\pune$\% of the  potential non-certain unexplained tuples, which we add to~$J$, and $\perr$\% of the potential non-certain error tuples, which we delete from~$J$.


\section{Mapping selection is NP-hard}

We provide a proof for the complexity result stated in Section III-C of the main paper.

\begin{theorem}
The mapping selection problem for full st tgds as defined in~Eq.~(4) of \cite{thispaper} is NP-hard.
\end{theorem}
\begin{proof}
We use a reduction from SET COVER, which is well known to be NP-complete, and is defined as follows:

\noindent 
$\mathbf{Given}$ a finite set $U$, a finite collection $R=\{R_i~|~R_i\subseteq U, 1\leq i\leq k\}$ and a natural number $n\leq k$, is there a set $R'\subseteq R$ consisting of at most $n$ sets $R_i$ such that $\bigcup_{R_i\in R'}R_i = U$?

We first consider the decision variant of mapping selection, which is defined as follows:

\noindent 
$\mathbf{Given}$ schemas \textbf{S}, \textbf{T}, a data example $(I,J)$, a set $\cand$ of candidate \emph{full} \tgd{}s, and a natural number $m$, is there a selection $\map \subseteq \cand$ with $F(\map)\leq m$?

\noindent 
where $F(\map) $ is the function minimized in Eq.~(4) of \cite{thispaper}, i.e.,  
\begin{align}
F(\map) =&\sum_{\q\in J} [1-\expf(\map,\q)] \nonumber\\
&+ \sum_{\q\in K_\cand-J}[\errf(\map,\q)] + \mathsf{size_m}(\map)
\end{align}

We construct a mapping selection decision instance from a SET COVER instance as follows. 
We set $m = 2n$, introduce an auxiliary domain 
 $D=\{1,\ldots,m+1\}$, and define
\begin{align*}
\mathbf{S} &= \{R_i/2~|~R_i\in R\}\\
\mathbf{T} &= \{U/2\}\\
\cand &= \{R_i(X,Y) \rightarrow U(X,Y)~|~R_i\in R\}\\
J &= \{U(x,y)~|~(x,y)\in U\times D\}\\
I &= \bigcup_{R_i\in R}\{R_i(x,y)~|~(x,y)\in R_i\times D\}
\end{align*}
It is easily verified that this construction is polynomial in the size of the SET COVER instance. 
We next show that the answers to SET COVER and the constructed mapping selection problem coincide.

For each $R_i$, the candidate st tgd $\theta_i = R_i(X,Y) \rightarrow U(X,Y)$ has size two, makes no errors (as $R_i\subseteq U$), and for each $x\in R_i$ explains the tuples $U(x,1),\ldots,U(x,m+1)$. We thus have
\begin{align}
F(\map) &=\sum_{\q\in J} [1-\expf(\map,\q)] + 2\cdot |\map|\\
&=(m+1)\cdot \left(|U|-|\bigcup_{\theta_i\in \map}R_i|\right) + 2\cdot |\map|
\end{align}
A mapping $\map\subseteq\cand$ with $F(\map)\leq m = 2n$ thus exists if and only if $|\bigcup_{\theta_i\in \map}R_i|=|U|$ and $|\map|\leq n$, which is exactly the case where $\map$ encodes a covering selection with at most $n$ sets. Furthermore, if such mappings exist, the optimal mapping according to~Eq.~(4) of \cite{thispaper} is one of them, and a polynomial time solution for mapping selection with full st tgds can thus be used to find a candidate solution that can be verified or rejected in polynomial time to answer SET COVER. 
\end{proof}

We note that the mapping selection problem for arbitrary st tgds as defined in~Eq.~(9) of \cite{thispaper} coincides with the one in~Eq.~(4) of \cite{thispaper} if all candidates are full, and thus is NP-hard as well. Furthermore, the reduction used in the proof directly generalizes to the following weighted version of the optimization criterion:
\begin{align*}
F(\map) =& w_1\cdot\sum_{\q\in J} [1-\expf(\map,\q)] \nonumber\\
&+ w_2\cdot\sum_{\q\in K_\cand-J}[\errf(\map,\q)] + w_3\cdot \sum_{\theta\in\map}\mathsf{size}(\theta)
\end{align*}
with positive integer weights $w_1, w_2, w_3$ and any size function that assigns equal size to the candidate mappings $\theta_i = R_i(X,Y) \rightarrow U(X,Y)$. More precisely, setting $m = \mathsf{size}(\theta_1) \cdot w_3 \cdot n$ in the proof above shows that this generalization is NP-hard as well.




\bibliographystyle{IEEEtran}
\bibliography{IEEEabrv,mapping}
%

\end{document}